\documentclass[11pt]{article}

\newif\ifanon
\anonfalse

\usepackage[american]{babel}
\usepackage[T1]{fontenc}
\usepackage[utf8]{inputenc}
\IfFileExists{lmodern.sty}
  {\usepackage{lmodern}\usepackage{microtype}}
  {\usepackage[expansion=false]{microtype}}
\usepackage[margin=1in]{geometry}
\usepackage{setspace}
\usepackage{lineno}
\usepackage{amsmath,amssymb,amsthm,mathtools}
\usepackage{booktabs,tabularx,array}
\usepackage{graphicx}
\usepackage[font=small,labelfont=bf]{caption}
\usepackage{enumitem}
\usepackage{xcolor}
\usepackage{xurl}
\usepackage[round,authoryear]{natbib}
\usepackage[colorlinks=true,allcolors=blue!55!black]{hyperref}

\bibpunct{(}{)}{;}{a}{}{,}

\modulolinenumbers[5]
\setlist{nosep,leftmargin=*}
\newcolumntype{Y}{>{\raggedright\arraybackslash}X}
\newtheorem{proposition}{Proposition}

\newcommand{\Ber}{\operatorname{Ber}}
\newcommand{\KL}{D_{\mathrm{KL}}}
\newcommand{\TV}{d_{\mathrm{TV}}}
\newcommand{\ind}{\mathbf{1}}

\newcommand{\orcid}[1]{\href{https://orcid.org/#1}{\nolinkurl{https://orcid.org/#1}}}

\title{Watermarking: from Impossibility to Auditable Compliance}
\ifanon
  \author{}
\else
  \author{Fernando Delbianco\textsuperscript{1,2}\and
           Fernando Tohm\'e\textsuperscript{1.2}\and
           Hugo Acciarri\textsuperscript{3}}
\fi
\date{}

\begin{document}
\maketitle
\ifanon\else
\begin{center}
\small
\textsuperscript{1}\,Departamento de Econom\'ia, Universidad Nacional del Sur (UNS), Bah\'ia Blanca, Argentina\\
\textsuperscript{2}\,Instituto de Matem\'atica de Bah\'ia Blanca (INMABB), CONICET--UNS, Bah\'ia Blanca, Argentina\\
\textsuperscript{3}\,Departamento de Derecho, Universidad Nacional del Sur (UNS), Bah\'ia Blanca, Argentina\\[0.8em]
\textbf{Corresponding author:} Fernando Delbianco, \href{mailto:fernando.delbianco@uns.edu.ar}{fernando.delbianco@uns.edu.ar}\\[0.5em]
\textbf{E-mail:} F.\ Delbianco: fernando.delbianco@uns.edu.ar\\
F.\ Tohm\'e: ftohme@criba.edu.ar\\
H.\ Acciarri: hacciarri@criba.edu.ar\\[0.5em]
\textbf{ORCID:} F.\ Delbianco: \orcid{0000-0002-1560-2587}\\
F.\ Tohm\'e: \orcid{0000-0003-2988-4519}\\
H.\ Acciarri: \orcid{0009-0005-8293-7087}
\end{center}
\fi
%\linenumbers

\begin{abstract}
Article 50 (2) of the EU Artificial Intelligence Act requires providers of generative systems to make synthetic outputs machine-readable and detectable, while qualifying the effectiveness, interoperability, robustness, and reliability by technical feasibility, cost, content-specific limits, and the state of the art. For free-form text, one important implementation route is the implementation of a generative watermarking procedure, which poses a compliance problem that is hard to address. Strong watermarking is impossible against adaptive removal, while ordinary edits attenuate statistical evidence, and unmarked human text may overlap distributionally with machine output. This article develops an auditable alternative. First, it defines a description-length robustness profile. A finite-sample bound shows that detectable bias decays and that the required sample size grows with the inverse square of the decay rate. This replaces an unidentified Shannon-entropy constant with {\em collision entropy}. Second, it constructs label-conditional conformal prediction sets with separate false-attribution and false-exclusion levels, reporting ``watermark supported,'' ``not supported,'' or ``inconclusive''. Coverage is obtained as a finite-sample result and is class-conditional under exchangeability. A small reproducible simulation of a tournament watermark confirms both claims and shows that the surviving-token rule overstates the tolerable edit rate roughly twofold. The resulting premarket certificate, signed detector report, and postmarket recalibration protocol operationalize the Commission's 2026 Code of Practice without claiming universal robustness.
\end{abstract}


\section{Introduction}

Text watermarking is one of the few techniques that can travel with free-form text after metadata is stripped, yet the same editability that makes text useful also makes a statistical watermark fragile. The strongest general guarantee, that no efficient adversary can remove a nontrivial mark while preserving output quality, is unavailable under broad, oracle-aided threat models \citep{zhang2024}. Even without an active adversary, the best possible classifier cannot reliably separate two source distributions that are close in total variation \citep{sadasivan2023}. At the same time, some evidence shows that generative watermarking can preserve measured quality and provide useful detection on sufficiently long text \citep{dathathri2024}. The practical question is consequently not whether a watermark is invulnerable, but which evidence survives which transformations, at what error rates, and under which documented conditions. This paper addresses this question under the emerging European regulatory regime.

This question became legally concrete on August~2, 2026. Article~50(2) of Regulation (EU) 2024/1689 requires providers of systems that generate synthetic audio, image, video, or text to ensure that outputs are machine-readably marked and detectable as artificially generated or manipulated \citep{euaiact2024}. The provision is not an absolute command. It requires technical solutions to be effective, interoperable, robust, and reliable ``as far as'' technically feasible, taking account of content limitations, implementation cost, and the generally acknowledged state of the art. The obligation does not apply in the cases that the AI system assists standard editing, without substantially altering the input data provided by the deployer or its semantics, or alternatively, when the AI system is authorized by law to be used to detect, prevent, investigate, or prosecute criminal offences. Regulation (EU) 2026/1744 allows for a transition until December~2, 2026 for systems placed on the market before August~2, 2026 \citep{eudigital2026}. The Commission's final Code of Practice and Article~50 Guidelines now give the provision an operational setting \citep{eucode2026,euguidelines2026}.

Article 50(2) requires providers to ensure that their marking and detection solutions are effective, interoperable, robust, and reliable insofar as technically feasible, making performance under relevant uses and transformations part of the compliance inquiry. The Code operationalizes that requirement for signatories through testing and documentation across expected contexts and transformations and, for free-form text, expressly recognizes that detection may sometimes yield comparatively less reliable, robust, or potentially misleading or low-confidence results. In such cases, the evidential limitations of the detector should be made explicit, and the result should not be presented as supporting a categorical conclusion where the available evidence does not justify one.. A single area under the receiver operating characteristic curve (AUROC) cannot answer the question. It averages across thresholds, obscuring a text's length and domain, and says little about a document that has moved outside the evaluation distribution.

This article proposes a joint technical and governance framework with five contributions:

\begin{enumerate}[label=(\roman*)]
    \item It defines \emph{watermark evidence bits}, interpreted as a minimum-description-length (MDL) device that yields auditable robustness.
    \item It derives a conservative bound for surviving edits, diluting the observed watermark. 
    \item It supplies a two-sided conformal procedure. Separate error levels control false attribution of the target watermark and false exclusion of genuinely marked text. Ambiguous and anomalous cases are reported as inconclusive.
    \item It reports a simulation of a watermark. It recovers the survival law and the corrected frontier to within Monte Carlo error. The surviving-token approximation overstates the tolerable substitution rate by a factor of nearly two.
    \item It maps these quantities to a certificate, a detector report, and a monitoring loop aligned with Article 50(2)'s criteria of effectiveness, reliability, robustness, and interoperability.
\end{enumerate}

The framework delivers less than a universal detector. The outcome is \emph{text carrying the watermark of a specified provider, model or sampler version, detector version, and key epoch}. Its complement includes human text, output from other models, provider output generated without the target mark, and text from which the mark has been removed. Accordingly, a negative result does not establish human authorship, and a positive result does not establish who submitted, edited, or endorsed the text. Those distinctions are essential since the output may affect a person's rights or reputation.

\section{Regulatory and Technical Background}

Three features of the current legal setting matter for a technical design. First, Article~50(2) states a provider obligation at the system-output layer. It is distinct from the disclosure duties imposed on deployers in other paragraphs. Second, its four performance attributes are expressly contextual and technologically qualified. A technical certificate may therefore delimit supported content types, lengths, languages, transformations, and operating points. Third, compliance is evaluated within a developing ecosystem of standards, codes, and supervisory practice rather than by a single statutory threshold.

The regulatory setting was subsequently supplemented by three separate instruments with different legal effects. 
The Code was the first of these three instruments to reach its final form on 10 June 2026, although the Commission’s work on the Article 50 Guidelines had begun earlier, with a public draft released on 8 May 2026. The Code is a voluntary compliance instrument rather than mandatory legislation. It translates the general requirements of Article 50(2), together with the obligations in Article 50(4) and (5), into more detailed operational commitments concerning matters such as marking techniques, testing, documentation, detector access, and interoperability. On 8 July 2026, the Commission assessed the Code as adequately covering the obligations in Article 50(2), (4), and (5), an assessment followed by the AI Board on 9 July. Article 50(7), as amended by Regulation (EU) 2026/1744, expressly provides for the Commission to assess whether adherence to the Code is adequate to ensure compliance with the obligations laid down in paragraphs 2 and 4, while the Code itself continues to address the horizontal requirements of paragraph 5. Adherence does not constitute conclusive evidence of compliance, nor does it give rise to a presumption of conformity of the kind that Article 40 attaches, within its own scope, to compliance with harmonized standards. The Code nevertheless occupies a privileged practical position. The Commission Guidelines describe it as “the only Union-wide recognised practical framework for demonstrating compliance” with the obligations in Article 50(2), (4), and (5). Non-signatories may demonstrate compliance through other adequate means, but those alternatives are not equivalent in practical terms: providers must independently explain and document how their measures satisfy Article 50, may have to carry out a gap analysis comparing those measures with the Code, and may face additional requests for information or access from competent authorities.
A second development was Regulation (EU) 2026/1744 (the AI Omnibus), adopted by the European Parliament and the Council on 8 July 2026. This instrument is binding legislation and amended the AI Act itself. For present purposes, its principal effect on our subject was to introduce a limited transitional rule for Article 50(2). Providers of systems generating synthetic audio, image, video, or text that had already been placed on the market before 2 August 2026 must take the necessary steps to comply with Article 50(2) by 2 December 2026. The four-month transition concerns only this obligation and this category of pre-existing systems; the other transparency obligations in Article 50 remain subject to their ordinary application dates (European Parliament and Council 2026).
The third instrument is the European Commission’s Guidelines on the implementation of Article 50, adopted and published on 20 July 2026 under Article 96(1)(d) of the AI Act (European Commission 2026b). The Guidelines are an official interpretive instrument. They explain how the Commission understands the scope and application of the binding obligations contained in Article 50, while leaving the Regulation itself unchanged. For Article 50(2), they clarify, among other matters, the treatment of source code, outputs intended exclusively for machine-to-machine communication, certain outputs used only within closed-loop industrial or product-development environments, standard editing, and narrowly defined business-to-business or industrial contexts. They therefore help identify when the statutory obligation applies without creating an independent marking or detection regime.
For the technical framework developed here, the resulting structure is consequential. Article 50(2) supplies the binding legal standard: qualifying outputs must be machine-readably marked and detectable through solutions that satisfy the Regulation’s contextual requirements of effectiveness, interoperability, robustness, and reliability. The Code supplies a more detailed, voluntary implementation route for signatories. The Guidelines are intended to clarify the scope and application of the underlying legal obligation. Regulation (EU) 2026/1744 modifies the binding framework itself, including the limited transition for systems already on the market. These instruments perform different functions and should therefore remain analytically separate.
This structure also leaves room for uncertainty at the level of individual detection results. Article 50(2) requires a technically adequate marking and detection system under the conditions specified by the Regulation. Conversely, it does not establish that every individual output must permit categorical attribution.

The final Code of Practice recognizes that no single marking technique generally satisfies all four Article~50(2) attributes across modalities. It adopts a multilayer approach, combining digitally signed metadata with content that can contain imperceptible watermarking. Free-form text cannot reliably retain container metadata, so the Code treats it as one of two narrow exceptions in which a single marking layer suffices. For free-form text, the Code requires signatories to apply an imperceptible watermark to texts longer than 200 tokens, while recognizing that reliability at such lengths may be lower than for very long texts. At the detection stage, it allows access to free-form-text watermark detectors to be restricted where comparatively lower reliability or robustness may produce misleading or low-confidence results, provided that useful access remains available, in particular to verified expert users with a legitimate need, and that the restriction lasts only until sufficiently reliable and robust detection mechanisms become established. The Code also calls for:

\begin{itemize}
    \item reliability testing across content length, size, diversity, semantics, and expected use contexts,
    \item robustness testing against ordinary processing, including lexical substitution, character insertion or deletion, paraphrasing, and translation cycles, and against adversarial attacks,
    \item interoperable access to detection mechanisms, with watermark-detection interoperability required by February~2, 2027, and
    \item downloadable, digitally signed detection results carrying identifying information.
\end{itemize}

These requirements create a need for reproducible claims, bounded scope, auditable error rates, and corrective action. The Commission Guidelines further explain that adherence to the voluntary Code can demonstrate compliance, while a provider using an alternative route must show that its means are adequate \citep{euguidelines2026}. The proposal below involves a measurement layer that can support either route.

The Code’s 200-token benchmark deserves particular attention because it links a specific implementation commitment for signatories to a statistical instrument whose evidential power varies with text length. Section~\ref{sec:demo} shows that, in a controlled simulation, text in the vicinity of that benchmark may still be too short to support a categorical conclusion once ordinary editing is introduced. The point at which the Code requires signatories to apply watermarking should therefore be distinguished from the point at which the resulting evidence becomes sufficiently informative for a categorical attribution. The certificate proposed in this article is designed to make that distinction explicit.

AI-text classifiers detect distributional features associated with machine output. Generative watermarks, instead, alter the sampling process using a secret or provider-controlled key and test for the induced statistical pattern \citep{kirchenbauer2023}. Cryptographic and coding-based work has produced distortion-free, undetectable, and provably robust constructions under different assumptions \citep{kuditipudi2024,christ2024,zhao2024,golowich2024,qu2025}. These are not interchangeable guarantees. ``Undetectable'' may mean computational indistinguishability to an observer without the key, whereas ``detectable'' means recoverable evidence for the authorized detector. ``Robust'' is meaningful only relative to the properties of the generative process.

More precisely, let $Y=1$ denote the fact that the input contains the target provider's active watermark under a declared configuration. Let $Y=0$ denote its complement. The hypotheses are
\begin{equation}
H_1:Y=1 \qquad\text{versus}\qquad H_0:Y=0.
\end{equation}
Neither hypothesis is equivalent to ``AI authored'' versus ``human authored.'' Mixed authorship, retrieval, quotations, provider-side exemptions, and post-generation edits all break that equivalence. 

\citet{zhang2024} show that strong watermarking is impossible in general against an efficient adversary equipped with suitable quality and perturbation oracles. But their result does not claim that every deployed watermark is useless. It just shows that a nontrivial, quality-preserving, universally removal-resistant guarantee cannot be assumed under that threat model. Later work sharpens robustness boundaries in coding abstractions and constructs guarantees for specified edit channels \citep{francati2025,golowich2024,qu2025,zhai2026}. 

\citet{sadasivan2023} establish a different, statistical boundary. For human-text distribution $P_H$, machine-text distribution $P_M$, and any detector $D$, and total variation distances $d_{TV}$,\footnote{A ROC curve illustrates the performance of a binary classifier model at varying threshold values. AUROC is the area under the ROC curve.}
\begin{equation}
\operatorname{AUROC}(D)
\leq \frac{1}{2}+\TV(P_M,P_H)-\frac{1}{2}\TV(P_M,P_H)^2.
\label{eq:tv-auroc}
\end{equation}
As the two distributions approach one another, the upper bound approaches a purely random valuation. A watermark can separate its marked distribution from the unmarked distribution, but editing can close that gap, and quality constraints limit how much separation a provider can introduce. Equation~\eqref{eq:tv-auroc} is therefore a reason to report the tested distributions and transformations.

The human costs of misclassification reinforce this requirement. Generic AI detectors have exhibited severe error disparities for nonnative English writing \citep{liang2023}. A provider-specific watermark detector avoids some feature-based biases, but length, language, topic, etc., can still change its score distribution.

SynthID-Text demonstrates that a generative watermark can operate at production scale. Its evaluation included approximately 20 million live Gemini responses and found no statistically significant degradation in the reported feedback measures. It also documented that longer and higher-entropy text carries more evidence and that editing or paraphrasing weakens detection \citep{dathathri2024}. The system uses pseudorandom scores and tournament sampling, generally with a four-token seed window and multiple tournament layers. It also offers selective prediction. 

Confidence intervals for an average benchmark metric do not become coverage guarantees for a later item. For that, conformal calibration is needed. This has already been applied in the watermark literature. \citet{xie2026} use conformal $p$-values to control false positives for adaptive AI-use policies in classrooms. The present contribution differs in purpose and construction. It addresses provider-side Article~50 auditing, ties calibration cells to an explicit robustness profile, and uses two label-conditional $p$-values so that both false attribution and false exclusion can be assigned separate tolerances. 

\section{A Description-Length Robustness Profile}

Let $\mathcal V$ denote the model's token vocabulary, a finite set of tokens its tokenizer can emit, fixed by the declared model and tokenizer version. Let $X_t \in \mathcal V$ be the token generated at position $t$ and $X_{<t}=(X_1,\ldots,X_{t-1})$ the vector of tokens preceding it. Consider an autoregressive model with token distribution
\begin{equation}
p_t(v)=\Pr(X_t=v\mid X_{<t}), \qquad v\in\mathcal V.
\end{equation}
A context-seeded watermark derives a pseudorandom seed\footnote{A deterministic map $
\operatorname{Hash}:K \times \mathcal{V}^h \rightarrow \mathcal{R}$ from a key and an $h$-token window to a seed space $\mathcal{R}$, whose outputs are computationally indistinguishable from independent uniform draws on 
$\mathcal{R}$ for any party not holding the key.}
\begin{equation}
r_t=\operatorname{Hash}(k,X_{t-h},\ldots,X_{t-1})
\end{equation}
from a secret key $k$ and a window of $h$ preceding tokens. A detector computes bounded contributions $Z_j\in[0,1]$ and aggregates them into a score $S$. Conditioning on the hypotheses, under a Bernoulli implementation, the stylized null and marked models are
\begin{equation}
Z_j\mid Y=0\sim\Ber(1/2),
\qquad
Z_j\mid Y=1\sim\Ber(1/2+\varepsilon),
\label{eq:ber-model}
\end{equation}
where $\varepsilon>0$ is the average marked bias. 

In a  single-layer tournament under the distribution $\Ber(1/2)$, let
\begin{equation}
C_t=\sum_{v\in\mathcal V}p_t(v)^2
\qquad\text{and}\qquad
H_2(p_t)=-\log C_t
\end{equation}
be the {\em collision probability} and the order-two R\'enyi entropy, respectively. The SynthID-Text analysis gives, after expectation over the random seed,
\begin{equation}
\Pr(Z_t=1\mid Y=1)=\frac{1}{2}+\frac{1-C_t}{4}.
\label{eq:collision-bias}
\end{equation}
The first-layer signal is $\varepsilon_t=(1-C_t)/4$, vanishing for a deterministic next-token distribution \citep{dathathri2024}. Deeper tournament layers alter the token distribution and its collision probability, so their contributions cannot simply be multiplied by the number of layers. Providers should report the actual covariance structure. If $n$ raw contributions have autocorrelations $r_q$, a conventional planning diagnostic is
\begin{equation}
n_{\mathrm{eff}}\approx
\frac{n}{1+2\sum_{q=1}^{q_{\max}}r_q},
\label{eq:neff}
\end{equation}
with a prespecified truncation. Substituting $n_{\mathrm{eff}}$ into an asymptotic formula is not a guarantee in itself but provides a useful heuristic. In Section~\ref{sec:demo}, Equation~\eqref{eq:neff} is applied on a simulated single-layer scheme, where it correctly returns $n_{\mathrm{eff}}\approx n$.

Only one edit model admits an exact survival calculation.It makes four assumptions: (i) Each token is substituted independently with probability $\delta$, (ii) the editor does not know the key, (iii) tokenization and length stay fixed. Finally, (iv) a contribution whose own token, or whose $h$-token seed context, has been touched loses its bias and behaves like the null in expectation.

Write $E_t=1$ if position $t$ was substituted and $E_t=0$ otherwise. Contribution $t$ survives only if none of the $h+1$ positions on which it depends has been touched, so the survival indicator is
\begin{equation}
U_t=\prod_{j=0}^{h}(1-E_{t-j}),
\end{equation}
with
\begin{equation}
\rho(\delta,h)=\mathbb E[U_t]=(1-\delta)^{h+1}.
\label{eq:survival}
\end{equation}
If the detector cannot identify which windows survived, the observed marked mean becomes diluted:
\begin{equation}
\mathbb E[Z_t\mid Y=1,\delta]
=\frac{1}{2}+\varepsilon\rho(\delta,h).
\label{eq:dilution}
\end{equation}

 Suppose a detector could discard every corrupted position. It would see about $n\rho$ contributions, each still biased by $\varepsilon$ so that the separation would depend on $n\varepsilon^2\rho$. A real detector would not be able to tell which windows survived. It only sees about $n$ contributions, each biased by  $\varepsilon\rho$, so the separation will depend on $n\varepsilon^2\rho^2$. The extra factor of $\rho$ is the whole difference between the two accounts.

\begin{proposition}
\label{prop:frontier}
Let $\bar Z$ be the mean of $n$ independent contributions taking values in $[0,1]$. Under $H_0$ each contribution has mean $1/2$. According to Equations~\ref{eq:survival} and \ref{eq:dilution}, each has mean $1/2+\varepsilon(1-\delta)^{h+1}$.  $H_0$ gets rejected when
\begin{equation}
\bar Z>\frac{1}{2}+\sqrt{\frac{\log(1/\alpha)}{2n}}.
\label{eq:hoeffding-test}
\end{equation}
Then the false-positive probability is at most $\alpha$. The false-negative probability is at most $\beta$ whenever
\begin{equation}
n\geq
\frac{\left(\sqrt{\log(1/\alpha)}+\sqrt{\log(1/\beta)}\right)^2}
{2\varepsilon^2(1-\delta)^{2(h+1)}}.
\label{eq:sample-bound}
\end{equation}
Equivalently, a sufficient maximum substitution rate is
\begin{equation}
\delta_{\max}^{\mathrm{H}}
=\left[
1-
\left\{
\frac{\sqrt{\log(1/\alpha)}+\sqrt{\log(1/\beta)}}
{\varepsilon\sqrt{2n}}
\right\}^{\!1/(h+1)}
\right]_+,
\label{eq:delta-hoeffding}
\end{equation}
where $[u]_+=\max(u,0)$.
\end{proposition}

\begin{proof}
Hoeffding's inequality applied to Equation~\eqref{eq:hoeffding-test} gives
$\Pr_{0}(\text{reject})\leq\alpha$. Under the edited marked hypothesis (Equations~\ref{eq:survival} and \ref{eq:dilution}), set
$\Delta=\varepsilon(1-\delta)^{h+1}$. The probability that $\bar Z$ falls below the threshold is at most
\begin{equation*}
\exp\!\left[-2n\left\{\Delta-
\sqrt{\frac{\log(1/\alpha)}{2n}}\right\}^2\right],
\end{equation*}
provided the expression in braces is nonnegative. Requiring this upper bound to be no larger than $\beta$ gives Equation~\eqref{eq:sample-bound}, and a rearrangement gives Equation~\eqref{eq:delta-hoeffding}.
\end{proof}

Hoeffding's inequality\footnote{Which yields an upper bound on the probability that the sum of bounded independent random variables deviates from its expected value by more than a certain amount.} holds for any bounded contribution. Under the Bernoulli model of Equation~\eqref{eq:ber-model} a normal approximation gives the tighter planning values
\begin{equation}
n\gtrsim
\frac{(z_{1-\alpha}+z_{1-\beta})^2}
{4\varepsilon^2(1-\delta)^{2(h+1)}}
\label{eq:normal-n}
\end{equation}
and
\begin{equation}
\delta_{\max}^{\mathrm{N}}
=\left[
1-
\left\{
\frac{z_{1-\alpha}+z_{1-\beta}}
{2\varepsilon\sqrt n}
\right\}^{\!1/(h+1)}
\right]_+.
\label{eq:normal-delta}
\end{equation}
Here $z_{1-\alpha}$ is the $(1-\alpha)$ quantile of the standard normal distribution.

Equation~\eqref{eq:normal-delta} shows the impact of the correction. A surviving-position derivation assumes that $n\rho$ contributions remain and that each keeps the full bias $\varepsilon$. That gives $1-\delta_{\max}\propto n^{-1/(h+1)}$. Because an unaided detector cannot find the surviving windows, the correct exponent is half as large: $1-\delta_{\max}\propto n^{-1/[2(h+1)]}$. The two rules disagree by a wide margin, and Section~\ref{sec:demo} measures the gap.

As an illustration only, Equation~\eqref{eq:normal-delta} with $\alpha=.01$, $\beta=.05$, $h=4$, and $\varepsilon=.05$ gives $\delta_{\max}^{\mathrm N}\approx.02$ at $n=2{,}000$ and $\approx.17$ at $n=10{,}000$. The difference of performance with SynthID stems from the assumptions (i)-(iv) of the edit model.

The {\em Minimal Description Length}(MDL) criterion selects the model that gives the shortest joint description of model and data \citep{rissanen1978,grunwald2007}. Let $P_1$ and $P_0$ be two fixed score models, with densities $f_1$ and $f_0$, describing the detector score under the marked and null hypotheses. For an observed score $s$, the codelength advantage of the marked model is the base-two log-likelihood ratio
\begin{equation}
b(s)=\log_2\frac{f_1(s)}{f_0(s)}.
\label{eq:item-bits}
\end{equation}

Its expectation under $P_1$ is
\begin{equation}
B(P_1,P_0)=\frac{\KL(P_1\Vert P_0)}{\log 2}
\quad\text{bits}.
\label{eq:evidence-bits}
\end{equation}
This provides an evaluation of the comparison of the two models. A detector report can disclose Equation~\eqref{eq:item-bits} for a single item, because it depends only on the observed score.

Under Equations~\eqref{eq:ber-model} and \eqref{eq:dilution}, the expected codelength advantage over $n$ independent contributions is
\begin{equation}
B(\delta)=\frac{n}{\log 2}
\KL\!\left(
\Ber\!\left\{\tfrac12+\varepsilon(1-\delta)^{h+1}\right\}
\middle\Vert \Ber(\tfrac12)
\right).
\label{eq:ber-bits}
\end{equation}
For small $\varepsilon(1-\delta)^{h+1}$, Equation~\eqref{eq:ber-bits} reduces to
\begin{equation}
B(\delta)\approx
\frac{2n\varepsilon^2}{\log 2}(1-\delta)^{2(h+1)}.
\label{eq:small-signal-bits}
\end{equation}
Equation~\eqref{eq:small-signal-bits} explains that, while edits reduce the mean signal linearly in $\rho$, the information separating two nearby Bernoulli models falls quadratically.

$B(\delta)$ is a \emph{description-length robustness curve}. That is, the loss of provider-specific statistical evidence under a declared transformation.

Let $S$ denote the aggregate detector score treated as a random variable, and let $F_{a,\theta}$ be a transformation family with operation type $a$ and severity $\theta$. The general profile is
\begin{equation}
B_{a,\theta,r}
=\frac{1}{\log 2}
\KL\!\left(
P_{S\mid Y=1,F_{a,\theta},R=r}
\middle\Vert
P_{S\mid Y=0,R=r}
\right),
\label{eq:generic-profile}
\end{equation}
where the random variable $R$ indexes calibration strata and $R=r$ fixes one of them. Providers must estimate Equation~\eqref{eq:generic-profile} on held-out data, report its uncertainty, and retain the underlying score distributions.

No single number captures every text change that matters legally or technically. Normalized token edit distance catches substitutions, insertions, deletions, and seed desynchronization, but it scores a meaning-preserving paraphrase as a large change. Embedding similarity tracks meaning better, but it can miss negation, changed named entities, and task-critical facts, and it depends on the embedding model chosen. The audit should therefore keep both:
\begin{align}
d_{\mathrm{tok}}(x,x')
&=\frac{\operatorname{Lev}_{\mathrm{tok}}(x,x')}
{\max\{|x|,|x'|\}},\\
d_{\mathrm{sem}}(x,x')
&=1-\operatorname{sim}_{E}(x,x'),
\end{align}
Here $x$ is the original text and $x'$ the transformed text, $\operatorname{Lev}_{\mathrm{tok}}$ is Levenshtein distance measured in tokens, $|x|$ is the token length of $x$, and $\operatorname{sim}_{\mathcal E}$ is cosine similarity under an embedding model $\mathcal E$. The script $\mathcal E$ distinguishes the embedding model from the expectation operator $\mathbb E$ and from the substitution indicator $E_t$. Both distances are preserved under the operation family $a$ (substitution, insertion or deletion, copy and paste, paraphrase, summarization, or translation), the language pair, and task-specific semantic checks, although the embedding model, the tokenizer, the normalization rules, and the thresholds must all be adapted under those transformations.

Proposition~\ref{prop:frontier} covers independent substitutions under fixed tokenization. Insertions and deletions change the text, as well as paraphrasing and translating. Those transformations belong to an empirical Pareto surface, capturing the trade-offs between them.

\section{Label-Conditional Conformal Calibration}
\label{sec:conformal}

Let $s(x)$ be the detector score assigned to a document $x$, increasing with evidence for the target watermark. It is the realized value of the random variable $S$, where the change of symbol marks the change of role, from a distribution to a fixed function to be calibrated. Once calibration begins, $s$ may not be retuned. Detector development data, calibration data, and robustness test data must be disjoint. Each example receives a provenance label $Y\in\{0,1\}$ established by controlled generation logs, not inferred by another detector. Calibration is performed within  cells $R=r$ that may combine:

\begin{itemize}
    \item provider, generator or sampler version, detector version, and key epoch,
    \item language, domain, and length bin,
    \item decoding settings and watermark-strength configuration, and
    \item allowed or expected edit family and severity band.
\end{itemize}

Where a cell holds only a few examples, it should be merged with others before calibration begins. Choosing a favorable cell {\em after} the test score has been revealed destroys the guarantee. Hierarchical pooling can improve efficiency, as in \citet{xie2026}, but the exchangeability or weighting assumptions must be recorded.

Define the nonconformity scores
\begin{equation}
A(x,1)=-s(x),
\qquad
A(x,0)=s(x).
\label{eq:nonconf}
\end{equation}
For a candidate label $y$ in cell $r$, let $I_{y,r}$ be the index for $n_{y,r}$ calibration examples with that label. The smoothed finite-sample conformal $p$-value is
\begin{equation}
p_y(x)=
\frac{1+\sum_{i\in I_{y,r}}
\ind\{A(X_i,y)\geq A(x,y)\}}
{n_{y,r}+1}.
\label{eq:pvalue}
\end{equation}

Let $p_1$ count the calibration items whose scores are no larger than $s(x)$, and $p_0$ count those whose scores are no smaller. Randomized tie-breaking may replace the $\geq$ relation where exact rather than conservative validity is wanted.

With separate tolerances $\alpha_0$ and $\alpha_1$, define
\begin{equation}
\Gamma(x)=\{y\in\{0,1\}:p_y(x)>\alpha_y\}.
\label{eq:predset}
\end{equation}
Attributing a watermark to a text that does not carry it and failing to find a mark that is there lead to different consequences. A regulator or user should then be able to price them separately, under two provisos: each cell requires a minimum calibration size, and the score may not be changed after the disputed item has been seen. 

\begin{proposition}
\label{prop:conformal}
Fix a cell $r$ and a label $y$. Suppose the $n_{y,r}$ calibration examples and a future example are exchangeable conditional on $Y=y$, $R=r$, and all model-training information. Then the set in Equation~\eqref{eq:predset} satisfies
\begin{equation}
\Pr\{y\notin\Gamma(X_{\mathrm{new}})\mid
Y_{\mathrm{new}}=y,R_{\mathrm{new}}=r\}
\leq \alpha_y.
\label{eq:coverage}
\end{equation}
\end{proposition}

\begin{proof}
Conditional exchangeability makes the rank of the future nonconformity score among the $n_{y,r}+1$ scores uniform when there are no ties and stochastically dominates the uniform one with the conservative tie rule. Therefore $\Pr\{p_y(X_{\mathrm{new}})\leq\alpha_y\mid Y=y,R=r\}\leq\alpha_y$. The event $y\notin\Gamma(X_{\mathrm{new}})$ is exactly this event.
\end{proof}

This is a finite-sample, class-conditional coverage statement \citep{vovk2005,vovk2013,angelopoulos2023}. Equation~\eqref{eq:coverage} averages over draws of the calibration set. For one \emph{particular} frozen calibration set, the realized error rate is itself random, with a standard deviation of order $\sqrt{\alpha_y(1-\alpha_y)/n_{y,r}}$. 

\begin{table}[t]
\centering
\caption{Permitted interpretation of the conformal prediction set. The target is a specified provider watermark, not human authorship.}
\label{tab:decisions}
\begin{tabularx}{\textwidth}{@{}lYY@{}}
\toprule
$\Gamma(x)$ & Reported result & Controlled error and required interpretation \\
\midrule
$\{1\}$ & Target watermark supported & If $Y=0$ and exchangeability holds in the declared cell, false attribution is bounded by $\alpha_0$. This supports a provider-specific mark. \\
$\{0\}$ & Target watermark not supported & If $Y=1$ and exchangeability holds, false exclusion is bounded by $\alpha_1$. It does not prove human authorship or absence of other AI assistance. \\
$\{0,1\}$ & Inconclusive & Both labels remain compatible with calibration data. No categorical source claim is warranted. \\
$\varnothing$ & Inconclusive; investigate drift & Both labels are atypical. Map to $\{0,1\}$ for decisions, preserve the anomaly flag, and examine version mismatch, attack, or distribution shift. \\
\bottomrule
\end{tabularx}
\end{table}

Table~\ref{tab:decisions} presents the four possible cases. Mapping the empty set to the full set can only enlarge predictions, so it preserves coverage.

The size of cells is a floor for what the levels can mean. The smallest attainable conformal $p$-value is $1/(n_{y,r}+1)$, so a nominal $\alpha_y=.01$ needs at least 99 calibration examples in that cell before exclusion is even possible. 

The two inconclusive cases arise in different regimes. Where the marked and null score distributions are well separated, a genuinely marked item in the lower $\alpha_1$ tail of the marked distribution is usually still far above the null. Both labels are then excluded, and the set is empty. Where separation is weak, as with short text, low entropy, or heavy editing, the same item fits both distributions and the set is full.

An empty set therefore means different things in different cells. In a well-separated cell it is the $\alpha_1$ ``budget'' being spent as expected. In a cell that normally returns full sets, it is an anomaly.

Validity does not survive repeated querying. If many detectors or many documents are tested and a decision maker will act on any positive result, the family-wise and false-discovery rates become uncontrolled. The audit policy must fix a correction or an error bound in advance. Repeatedly testing revised versions of one document raises the same problem, unless the sequence is treated as a family or handled by a sequentially valid method.

Exchangeability is crucial. There are several issues (a model update, a new key epoch, etc.) that can break it. Weighted conformal prediction can restore validity under a specified covariate shift, provided the density ratio is known or can be estimated consistently \citep{tibshirani2019}. More general departures yield only approximate bounds, depending on the assumptions adopted \citep{barber2023}. Every recalibration creates a new detector version and a new calibration hash. 

\section{A Synthetic Demonstration}
\label{sec:demo}

The preceding sections make three claims:
\begin{itemize}
\item evidence decays as $(1-\delta)^{h+1}$,
\item the tolerable substitution rate follows the corrected $\rho^{2}$ frontier,
\item a conditional conformal layer delivers class-conditional error rates, degrading into abstention. 
\end{itemize}

This section checks all three in a small simulator. The value of this exercise is that it makes the framework cheap to falsify. On top of that, it adds a fourth result, specifying what happens when a calibration cell is mis-specified. 

The simulator implements a miniature single-layer tournament watermark. A pseudorandom function with a key maps a packed $(h\text{-token context},\text{token})$ pair to a value in $\{0,1\}$ under a Bernoulli distribution. The generation process draws two candidate tokens from the next-token distribution and emits the one with the larger value in $\{0,1\}$, breaking ties with a keyed coin. Unmarked text is drawn from the same distributions, which gives exactly matched null data.

Each next-token distribution places mass $q^{\ast}$ on one token and spreads $1-q^{\ast}$ over $K-1$ alternatives. The collision probability is then available in closed form as $C_t=(q^{\ast})^2+(1-q^{\ast})^2/(K-1)$, and varying $(q^{\ast},K)$ we capture the entropy of the distribution. Two strata are used, stratum~A, with mean per-token bias $\bar\varepsilon=.095$, and the lower-entropy stratum~B, with $\bar\varepsilon=.047$. Here $\bar\varepsilon$ is the average over positions and documents of the per-token signal $\varepsilon_t$ of Equation~\eqref{eq:collision-bias}, and it plays the role of $\varepsilon$ in Section~3.

The edit channel substitutes each token independently with probability $\delta$, drawing the replacement uniformly from the vocabulary. Detection recomputes every value $Z_t \in \{0,1\}$ and reports the standardized score
\begin{equation}
s=\frac{\sum_t Z_t-n/2}{\sqrt{n/4}},
\label{eq:zstat}
\end{equation}
where $n$ is the number of scorable positions, as in Section~3. A document of $N$ tokens gives $n=N-h$ of them.

The null score has mean $1/2$ per contribution and unit variance after standardization. The marked bias equals the closed-form $\mathbb E[(1-C_t)/4]$ to within $.003$. And a substituted token or context yields a fresh pseudorandom bit, which is exactly the ``null-like corrupted window'' assumption of Proposition~\ref{prop:frontier}. Throughout, $h=4$, matching the seed-window length reported for SynthID-Text.

Equation~\eqref{eq:neff} returns $n_{\mathrm{eff}}$ between $409$ and $1{,}621$ against nominal $n$ of $396$ and $1{,}596$. Single-layer contributions with distinct contexts are independent, something that cannot be assumed for multilayer tournaments.

\begin{figure}[t]
\centering
\includegraphics[width=\textwidth]{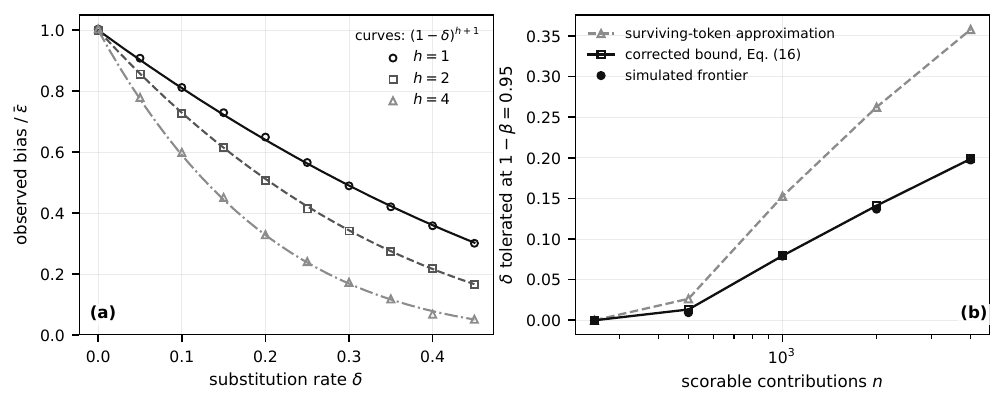}
\caption{Context survival and the corrected frontier. \textbf{a} Observed marked bias after substitution, normalized by the unedited bias $\bar\varepsilon$. Simulation values are the marks on the curves representing $(1-\delta)^{h+1}$. \textbf{b} Substitution rate tolerated at $\alpha=.01$, $1-\beta=.95$ in stratum~A as a function of the number of scorable contributions. The surviving-token approximation allows roughly twice the edit rate the detector can in fact tolerate.}
\label{fig:frontier}
\end{figure}

Figure~\ref{fig:frontier}a plots the observed marked bias after substitution, normalized by the unedited bias, against $\delta$ for $h\in\{1,2,4\}$. Simulation and the analytic curve $(1-\delta)^{h+1}$ differ by at most $.010$ over $\delta\in[0,.45]$, with mean deviations below $.005$. Equations~\eqref{eq:survival} and~\eqref{eq:dilution} therefore describe this channel accurately, including the exponent $h+1$.

Figure~\ref{fig:frontier}b tests the correction itself. For each length, the largest $\delta$ at which the one-sided test keeps power $1-\beta=.95$ at $\alpha=.01$ was estimated by simulation, then compared with two predictions: Equation~\eqref{eq:normal-delta}, and the surviving-position rule that credits the detector with $n\rho$ contributions at bias $\varepsilon$.

The simulated frontier tracks Equation~\eqref{eq:normal-delta} closely: $.079$ against $.079$ at $n=1{,}000$, $.137$ against $.141$ at $n=2{,}000$, and $.197$ against $.199$ at $n=4{,}000$. The surviving-position rule is optimistic by a factor of $1.8$ to $2.0$ throughout. At $n=4{,}000$ in stratum~A it would certify tolerance to a $36$ percent substitution rate where the true limit is $20$ percent. Stratum~B repeats the pattern at roughly four times the length, which is what Equation~\eqref{eq:normal-n} predicts from the ratio of squared biases.

A provider reasoning from surviving tokens would report an edit tolerance about twice the achievable one, and the error grows with the seed-window length. This is a claim a supervisory authority or an independent assessor could falsify using the transformation tests the Code already contemplates.

\begin{figure}[t]
\centering
\includegraphics[width=\textwidth]{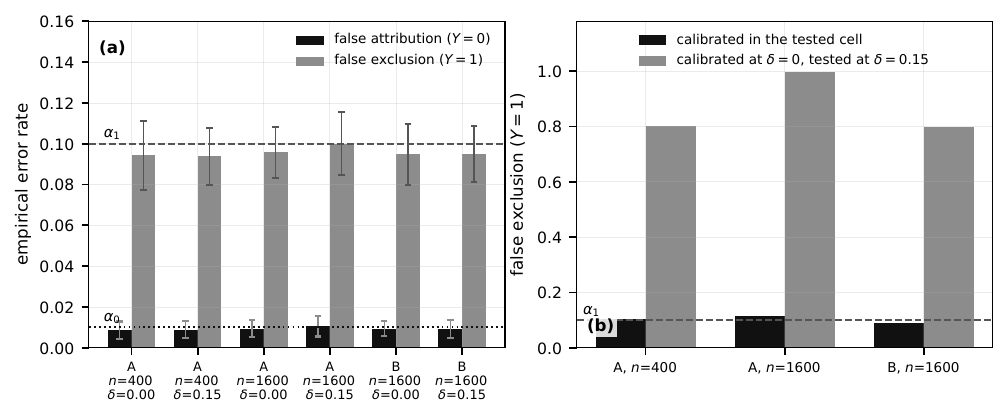}
\caption{Label-conditional conformal calibration with $\alpha_0=.01$ and $\alpha_1=.10$. \textbf{a} Empirical false attribution and false exclusion in six calibration cells, averaged over 60 calibration draws of 500 examples per label; error bars are the standard deviation across draws. \textbf{b} False exclusion when the item is scored against a cell calibrated on unedited text but is itself edited at $\delta=.15$, compared with calibration in the tested cell}
\label{fig:coverage}
\end{figure}

Six cells were calibrated at $\alpha_0=.01$ for false attribution and $\alpha_1=.10$ for false exclusion. The asymmetry is the one recommended by the framework, attributing a provider's watermark to text that does not carry it is the more consequential error. Table~\ref{tab:demo} reports the results. Empirical false attribution falls between $.0086$ and $.0105$, and empirical false exclusion between $.0938$ and $.1000$, against nominal levels of $.01$ and $.10$.

Figure~\ref{fig:coverage}a also shows the dispersion anticipated in Section~\ref{sec:conformal}. Across 60 calibration draws of $500$ examples per label, the realized false-exclusion rate has a standard deviation of $.013$ to $.017$, and the realized false-attribution rate one of $.004$ to $.005$. Both match $\sqrt{\alpha_y(1-\alpha_y)/n_{y,r}}$ to two significant figures. The guarantee is an average over calibration draws, and any single deployed calibration object sits somewhere in that spread. A certificate reporting only the nominal level hides this.

Figure~\ref{fig:coverage}b shows what a mis-specified cell costs. Items edited at $\delta=.15$ were scored against a cell calibrated on unedited text. A provider calibrates once on a clean output, then applies the result to text found in the wild. False exclusion rises from about $.10$ to between $.80$ and $.998$.

The detector reports ``watermark not supported'' for nearly every genuinely marked document. This is what the exchangeability requirement amounts to in practice, and it is the strongest argument for three things: 
\begin{itemize}
\item declaring the edit band as part of the cell, 
\item carrying an in-scope flag on the detector report, and 
\item answering out-of-scope items inconclusive rather than confidently.
\end{itemize}

\begin{figure}[t]
\centering
\includegraphics[width=\textwidth]{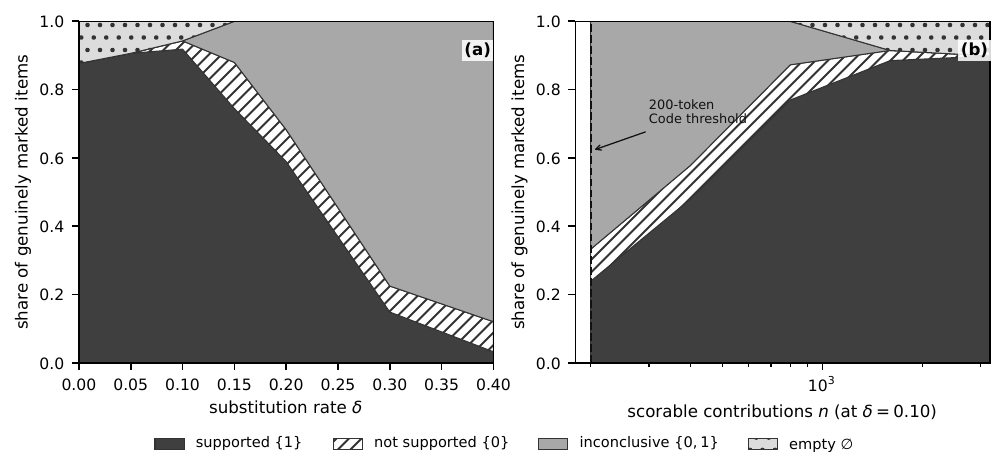}
\caption{Decision composition for genuinely marked items at $\alpha_0=.01$, $\alpha_1=.10$, using a single calibration draw of 800 examples per label. \textbf{a} By substitution rate at $n=1{,}600$. \textbf{b} By length at $\delta=.10$, with the Code of Practice's 200-token threshold marked. Evidence loss is absorbed by the inconclusive categories rather than by confident misclassification}
\label{fig:abstention}
\end{figure}

Figure~\ref{fig:abstention}a follows genuinely marked documents as editing increases at fixed length. The share reported ``watermark supported'' falls from $.877$ at $\delta=0$ to $.033$ at $\delta=.40$. The displaced mass moves into the inconclusive categories, which rise from $.123$ to $.879$. The share reported ``watermark not supported'' stays within the $\alpha_1$ budget throughout, and false attribution on unmarked items never exceeds $.017$.

As the marked and null distributions converge, which is the regime Equation~\eqref{eq:tv-auroc} describes, the system loses the ability to make any claim before it starts making false ones.

The figure also shows the transition between the two inconclusive categories described in Section~\ref{sec:conformal}. At low $\delta$, separation is strong and the $\alpha_1$ budget goes on empty sets. From $\delta=.15$ the marked and null calibration distributions overlap, and the budget goes on full sets instead. Both are reported as inconclusive.

Figure~\ref{fig:abstention}b has the most direct regulatory bearing. At a modest $\delta=.10$, the share of genuinely marked documents reported as supported rises from $.239$ at $n=200$ to $.899$ at $n=3{,}200$. At a 200-token threshold, two-thirds of the genuinely marked documents are inconclusive in this simulation.

The benchmark provides a Code implementation rule for signatories as to when watermarking should be applied. A compliance regime that reads the marking obligation as an implied detection warranty will overestimate what the evidence supports for short texts. Those are exactly the texts most likely to be disputed (single messages, a review, or an answer).

\begin{table}[t]
\centering
\caption{Simulated calibration cells. Levels are $\alpha_0=.01$ for false attribution and $\alpha_1=.10$ for false exclusion; entries are means over 60 draws of 500 calibration examples per label, with the standard deviation across draws in parentheses. Abstention is the combined share of inconclusive outputs across both labels}
\label{tab:demo}
\small
\begin{tabular}{@{}llrrrrr@{}}
\toprule
Stratum & $\bar\varepsilon$ & $n$ & $\delta$ &
False attribution & False exclusion & Abstention \\
\midrule
A & .095 & 400 & .00 & .0086 (.0042) & .0943 (.0169) & .017 \\
A & .095 & 400 & .15 & .0089 (.0043) & .0938 (.0139) & .515 \\
A & .095 & 1{,}600 & .00 & .0094 (.0041) & .0959 (.0125) & .053 \\
A & .095 & 1{,}600 & .15 & .0105 (.0050) & .1000 (.0154) & .049 \\
B & .047 & 1{,}600 & .00 & .0094 (.0037) & .0948 (.0149) & .014 \\
B & .047 & 1{,}600 & .15 & .0093 (.0045) & .0950 (.0138) & .519 \\
\bottomrule
\end{tabular}
\end{table}

Table~\ref{tab:demo} shows that abstention is not monotone in any single variable. Stratum~A at $n=1{,}600$ and $\delta=.15$ abstains on $.049$ of items. The same length and edit band in stratum~B abstains on $.519$. Two cells with identical declared length and identical declared edit severity behave completely differently, because their entropy differs. A certificate that reports length and edit tolerance but not the entropy regime is therefore incomplete.

The simulated next-token distributions are parametric and independent across positions, so they lack the long-range structure of natural language. Real per-token bias is heterogeneous in ways this construction only crudely imitates. The edit channel is analytically tractable, but the demonstration does not extend to insertion, deletion, paraphrase, or translation. Only one tournament layer is simulated, so the dependence that motivates Equation~\eqref{eq:neff} never arises. And because calibration and test data come from one process, exchangeability holds by construction except where Figure~\ref{fig:coverage}b  breaks it.

\section{An Auditable Article~50 Architecture}

The description-length profile and the conformal set do different jobs. The profile characterizes how evidence degrades across transformations. The conformal set converts a frozen score into set-valued decisions with bounded error for a specified population. Our proposed compliance architecture therefore combines both within a lifecycle.

\begin{itemize}
\item \textbf{Proposed predeployment robustness certificate.}
Before deployment, the provider chooses the generator or sampler, the detector, the tokenizer, the key-epoch, the policy, the target populations, the transformations, and the acceptance criteria. Independent held-out data then produce a certificate that includes negative findings and unsupported cells. Notice that this certificate does not constitute an EU conformity certificate, a conformity assessment, or a presumption of compliance. 

\item \textbf{Signed per-query detector report.}
Each detector response identifies the exact system and calibration object used, records both conformal $p$-values and the prediction set, and cryptographically binds the result to a hash of the tested bytes. It reports whether the item lies within a certified cell and implements the Code's call for downloadable, signed detection output.

\item \textbf{Lifecycle monitoring and recalibration.}
Providers monitor score distributions, edit channels, detector abuse, false-attribution reports, and model changes. A material generator, tokenizer, detector, or calibration change creates a new version. This lifecycle is consistent with the testing and maintenance approach recommended for synthetic-content transparency \citep{nist2024}.
\end{itemize}

Table~\ref{tab:certificate} lists the fields a certificate should contain and what each is for. The population row must record the entropy regime as well as the length bin, because Table~\ref{tab:demo} shows two cells with the same length and edit severity behaving entirely differently when entropy differs. The robustness row must record the edit band used for calibration, because Figure~\ref{fig:coverage}b shows that scoring an edited item against an unedited cell inverts the guarantee rather than merely loosening it.

\begin{table}[t]
\centering
\caption{Proposed certificate fields and their compliance function}
\label{tab:certificate}
\small
\begin{tabularx}{\textwidth}{@{}>{\raggedright\arraybackslash}p{0.19\textwidth}YY@{}}
\toprule
Field & Minimum disclosure & Function \\
\midrule
System identity & Provider, model and sampler versions, watermark mode, tokenizer, detector, key epoch or rotation policy & Prevents performance claims from migrating across incompatible versions. \\
Population & Languages, domains, length bins, entropy regime, decoding settings, exclusions, and intended uses & Defines the exchangeability claim and the scope of effectiveness. \\
Null and positive labels & Operational construction of $Y=0$ and $Y=1$, matched prompt protocol, provenance logs & Makes ``reliability'' a reproducible source-specific test. \\
Quality & Blind human preference where material, task scores, diversity, latency; uncertainty & Documents the quality--detectability trade-off  \\
Robustness surface & Operation family, token and semantic distances, calibrated edit band, attack knowledge, score/evidence bits, FPR, TPR, and abstention with intervals & Gives a bounded, multidimensional meaning to robustness. \\
Conformal calibration & Nonconformity function, cells, $n_{0,r}$, $n_{1,r}$, $\alpha_0$, $\alpha_1$, realized dispersion, hashes, and validity date & Makes the decision rule and finite-sample resolution auditable. \\
Interoperability & Authenticated API or standard access method; schemas, rate limits, detector discovery and retirement policy & Allows authorized third parties to reproduce detection without exposing keys. \\
Security and abuse & Key custody, detector-query controls, spoofing and removal tests, incident process & Addresses the possibility that detector access itself enables attacks. \\
Monitoring & Drift metrics, thresholds, review interval, corrective actions, and version history & Converts a one-time benchmark into continuing assurance. \\
\bottomrule
\end{tabularx}
\end{table}

The certificate should expose enough aggregate information for scrutiny without publishing secret keys or an attack-enabling transcript. A competent independent assessor can receive further confidential material.

A machine-readable report should contain at least:

\begin{enumerate}
    \item a cryptographic hash of the exact input bytes and the normalization/tokenization procedure,
    \item provider, generator or sampler, watermark, detector, and calibration-set identifiers,
    \item timestamp, key epoch tested, declared cell, input length, and in-scope/out-of-scope status,
    \item raw detector score, $p_0$, $p_1$, $\alpha_0$, $\alpha_1$, $n_{0,r}$, $n_{1,r}$, and $\Gamma(x)$;
    \item the three-state natural-language conclusion from Table~\ref{tab:decisions};
    \item active drift or integrity warnings and a link to the applicable certificate version, and
    \item the detector service's digital signature and verification information.
\end{enumerate}

The report must never present $1-p_0$, the raw score, or a classifier confidence as ``the probability that AI wrote this.''  Conformal coverage does not allocate a legal burden of proof, which is supplied by the applicable substantive and procedural law and applied by the competent decision-maker.

Table~\ref{tab:article50} translates each of the four attributes named in Article~50(2) into evidence a provider can produce.

\begin{table}[t]
\centering
\caption{Operational interpretation of Article~50(2)'s four technical attributes}
\label{tab:article50}
\begin{tabularx}{\textwidth}{@{}p{0.16\textwidth}YY@{}}
\toprule
Attribute & Proposed evidence & What it does not establish \\
\midrule
Effectiveness & Coverage of in-scope outputs, detectable-length distribution, TPR and abstention at declared $\alpha_0,\alpha_1$ & Universal marking of exempt, truncated, unsupported, or externally generated content \\
Reliability & Label-conditional error control in declared cells, subgroup results; signed reproducible reports & Pointwise certainty for a disputed document or validity after unmeasured shift \\
Robustness & Description-length and performance surfaces across named ordinary and adversarial transformations & Resistance to every adaptive, quality-preserving removal strategy \\
Interoperability & Versioned detector discovery, authenticated access, common schemas, signatures, and retirement notices & Semantic equivalence among different providers' scores or keys \\
\bottomrule
\end{tabularx}
\end{table}

\section{Discussion and Conclusion}

Text watermarking cannot carry universal artificial content detection. Quality-preserving robustness fails under adaptive threat models. Statistical separability vanishes as source distributions converge. And ordinary editing erases evidence faster than a surviving-token count suggests, by a factor of roughly two in the simulation reported here.

None of this makes Article~50 incoherent. The provision itself qualifies its technical attributes by feasibility, content limitations, cost, and the state of the art. The Code adds voluntary implementation commitments concerning layered marking, testing, signed detection results, and continuing improvement, while the Guidelines clarify the scope and application of the statutory obligations.

The response is to generate an auditable object. A description-length robustness profile states how much source-specific evidence survives a named transformation. Collision probability and empirical score distributions replace an unidentified entropy multiplier. Token edit distance and semantic loss stay on separate axes. A label-conditional conformal layer then turns the frozen detector score into a supported, unsupported, or inconclusive result, at separately chosen false-attribution and false-exclusion levels, under stated exchangeability conditions.

Three findings from the demonstration carry straight into practice. First, a robustness claim built on surviving-token reasoning overstates tolerable editing about twofold, and the overstatement grows with the seed's window length. Second, a calibration cell that omits the edit band does not degrade gracefully on edited text. It reports ``not supported'' for nearly every genuinely marked document. Third, assuming a 200-token threshold, editing already leaves most genuinely marked documents inconclusive. A duty to mark is not a warranty of detectability.

This architecture does not defeat the impossibility theorem, and it assigns no posterior probability to an individual text. Its value is that it makes claims falsifiable. Regulators gain a way to read ``robust and reliable as far as technically feasible'' as evidence comparable across providers and over time. Providers gain a way to disclose failure regions without treating every limitation as noncompliance. Watermarking works as an indicator of provenance rather than as an oracle of authorship.

A final caveat is that this design has the potential to create an incentives problem for providers, since lower reliability may justify more restrictive access while improvements in detector performance may reduce a provider’s regulatory discretion. The issue is in part endogenous, since the technical limitation that triggers the restriction is itself influenced by the provider’s choices concerning the development, validation, and deployment of the detector.

\section*{Declarations}

\paragraph{Data availability.}
This article develops an analytic and governance framework and reports no new empirical dataset, human-participant experiment, or proprietary system evaluation. All data underlying Section~\ref{sec:demo} are synthetic and are regenerated deterministically by the accompanying code from the seed recorded in Appendix~C; no separate data deposit is required.

\paragraph{Code availability.}
The simulator, experiment driver, and test suite described in Appendix~C are provided as supplementary material and will be deposited in a public repository with a persistent identifier on acceptance. The code depends only on \texttt{numpy}, \texttt{scipy}, and \texttt{matplotlib}, and reproduces every figure and every number in Section~\ref{sec:demo} in a few minutes on a single core.

\paragraph{Ethics.}
No human participants or personal data were involved. A future empirical validation should archive prompts, generation settings, transformation code, prespecified cells, calibration splits, and report schemas, subject to key-security, privacy, and ethics constraints. Any evaluation using user text or consequential decisions requires the applicable ethics, privacy, and due-process review.

\paragraph{Competing interests and funding.}
\ifanon
Provided on the separate title page in accordance with the journal's double-anonymous submission procedure.
\else
\textbf{Funding.} The authors received no specific funding for this work.
\textbf{Competing interests.} The authors have no relevant financial or non-financial interests to disclose.
\fi

\appendix

\section{Reproducing the Demonstration}
\label{app:code}

The supplementary code has three files. \texttt{wmdemo.py} implements the keyed pseudorandom function, the next-token model, tournament generation, the substitution channel, the detector, the conformal $p$-values of Equation~\eqref{eq:pvalue}, the decision rule of Equation~\eqref{eq:predset}, and the diagnostic of Equation~\eqref{eq:neff}. \texttt{run\_demo.py} runs the three experiments and writes the figures and a machine-readable \texttt{results.json}. \texttt{test\_wmdemo.py} asserts the properties the propositions assume: that the null score is unbiased and standard normal after standardization; that the marked bias equals $\mathbb E[(1-C_t)/4]$; that $\rho=(1-\delta)^{h+1}$ holds for $h\in\{1,3,5\}$; that both conformal $p$-values are super-uniform at the .01, .05, .10, and .20 levels; and that realized class-conditional errors respect $\alpha_0$ and $\alpha_1$.

\begin{verbatim}
    python3 test_wmdemo.py     # property checks, a few seconds
    python3 run_demo.py        # figures + results.json, a few minutes
\end{verbatim}

The master seed is 20260902 and the seed-window length is $h=4$. The planning levels are $\alpha=.01$ and $\beta=.05$; the operating levels are $\alpha_0=.01$ and $\alpha_1=.10$. Stratum~A uses $K\in\{2,4,8,16\}$ and $q^{\ast}\sim\mathrm{Unif}(.55,.97)$, stratum~B uses $K\in\{2,4,8\}$ and $q^{\ast}\sim\mathrm{Unif}(.80,.99)$. Every reported quantity is a deterministic function of these settings.

\end{document}